\documentclass[aps,pra,reprint,groupedaddress,superscriptaddress,longbibliography]{revtex4-2}
\pdfoutput=1
\usepackage{caption}
\usepackage{physics}
\usepackage{amsmath,amsthm}
\usepackage{graphicx}
\usepackage{bm}
\usepackage{bbm}
\usepackage{hyperref}
\usepackage{IEEEtrantools}

\DeclareMathOperator*{\minimize}{minimize:}

\newcommand{\Hspace}{\mathcal{H}}
\newcommand{\densityH}[1][\mathcal{H}]{\tilde{D}(#1)}
\newcommand{\normdensityH}[1][\mathcal{H}]{D(#1)}

\newcommand{\Hinf}{\ensuremath{\mathcal{H}_\infty}}

\newcommand{\im}{\textrm{im}}
\newcommand{\Pin}{\bar{\Pi}}

\newcommand{\scpub}{S_C} 
\newcommand{\keypriv}{S_Z} 
\newcommand{\skey}{S_K}

\newcommand{\rhot}{\tilde{\rho}}
\newcommand{\sigmat}{\tilde{\sigma}}

\newcommand{\sinf}{\ensuremath{\mathbf{S}_\infty}}
\newcommand{\sfin}{\ensuremath{\mathbf{S}_N}}
 
\newcommand{\e}{\epsilon}
\newcommand{\rrho}{\sqrt{\rho}}
\newcommand{\tnorm}[1]{\norm{#1}_1}
\newcommand{\snorm}[1]{\norm{#1}_\infty}
 
\newcommand{\id}{\mathbbm{1}}
\newcommand{\adj}[1]{{#1}^\dagger}

\newtheorem{thm}{Theorem}

\newtheorem{lemma}{Lemma}

\theoremstyle{definition}

\begin{document}

\title{An Improved Correction Term for Dimension Reduction in Quantum Key Distribution}

\author{Twesh Upadhyaya}
\email{twesh.upadhyaya@uwaterloo.ca}
\affiliation{Institute for Quantum Computing and Department of Physics and Astronomy\\University of Waterloo, Waterloo, Ontario, Canada N2L 3G1}
\author{Thomas van Himbeeck}
\affiliation{Institute for Quantum Computing and Department of Physics and Astronomy\\University of Waterloo, Waterloo, Ontario, Canada N2L 3G1}
\affiliation{Department of Electrical \& Computer Engineering, University of Toronto, Toronto, Ontario, Canada M5S 3G4}
\author{Norbert L\"{u}tkenhaus}
\affiliation{Institute for Quantum Computing and Department of Physics and Astronomy\\University of Waterloo, Waterloo, Ontario, Canada N2L 3G1}

\date{\today}

\begin{abstract}
The dimension reduction method \cite{Upadhyaya2021} enables security proofs of quantum key distribution (QKD) protocols that are originally formulated in infinite dimensions via reduction to a tractable finite-dimensional optimization. The reduction of dimensions is associated with a correction term in the secret key rate calculation. The previously derived correction term is loose when the protocol measurements are nearly block-diagonal with respect to the projection onto the reduced finite-dimensional subspace. Here, we provide a tighter correction term. It interpolates between the two extreme cases where all measurement operators are block-diagonal, and where at least one has maximally large off-diagonal blocks. This new correction term can reduce the computational overhead of applying the dimension reduction method by reducing the required dimension of the chosen subspace.
\end{abstract}

\maketitle

\section{Introduction}
Quantum key distribution (QKD) is a promising quantum technology, enabling two parties to communicate securely, even if an eavesdropper has unlimited computational power \cite{Scarani2009,Xu2020,Pirandola2020}. 

A security analysis of a QKD protocol derives the rate at which a secret key can be generated at the specified security level. Recently, numerical tools have been introduced to perform these key rate calculations, by reliably solving an optimization over the joint state of Alice and Bob \cite{Coles2016,Winick2018}. These tools are useful for practical QKD security proofs as they enable modelling of imperfect devices, encapsulate extended side channel models, and take advantage of specific details of the observed data to get an increased key generation rate.

As most QKD protocols are implemented optically, the underlying Hilbert space is infinite-dimensional. The relevant optimization is then over infinite-dimensional states, so the aforementioned numerical tools cannot be applied directly. For discrete-variable (DV) protocols, techniques such as squashing maps or the more general flag-state squasher can be used to map the problem to an effective finite-dimensional optimization \cite{Zhang2021,Li2020}. These tools, however, are not straightforwardly applicable to continuous-variable (CV) protocols. This is because they rely on all the protocol's measurement operators commuting with a projector on an underlying low-dimensional subspace, a subspace which essentially captures the protocol's behaviour.

Recently, we have extended the numerical framework to encompass both CV and DV protocols in infinite-dimensional Hilbert spaces via the \emph{dimension reduction method} \cite{Upadhyaya2021}. This method provides a tight lower bound on the infinite-dimensional key rate optimization by optimizing a specified finite-dimensional problem, and subtracting a correction term that bounds the difference between the original infinite-dimensional problem and the finite-dimensional one. The correction term bounds how much the key rate can increase under projection. In our previous work, we found an analytic form for the correction term that was applicable to any QKD protocol, but loose in certain cases. We also found that the correction term is zero when all the POVMs commute with the same projector. We conjectured that a tighter correction term exists which interpolates between these two cases; becoming smaller when the measurements are closer to block-diagonal. In this work, we find such a correction term. Practically, this is relevant for improving the performance of the dimension reduction method, and enabling its applications to more computationally demanding scenarios. It may also be of independent interest to better understand how the key rate changes under projection.

\section{Background}
In this section we briefly review the formulation of the asymptotic key rate as a convex minimization and the dimension reduction method; focusing in particular on the correction term. For a more detailed discussion, we refer the reader to Ref. \cite{Upadhyaya2021}.

\subsection{Key Rate Optimization and Dimension Reduction Method}
In each key generation round of a QKD protocol, Alice and Bob establish a quantum state $\rho_{AB}$; and Eve holds its purification in her register $E$. Alice and Bob measure their respective subsystems and perform classical data processing, which may involve public announcements, to generate a raw key, which is the key before error correction and privacy amplification. The raw key is stored in the register $Z$, and any public announcements are stored in the register $C$. These measurement and postprocessing steps can be represented as a quantum-to-classical channel $\Phi: AB \rightarrow ZC$. We use the notation $[E]=EC$ for the composite register containing all information accessible to Eve. The asymptotic key rate per signal sent under collective attacks is given by the Devetak-Winter formula \cite{Devetak2005}, which can be expressed as a convex optimization \cite{Winick2018},
\begin{equation}
\label{convexkeyrate}
R^\infty=\min_{\rho_{AB} \in \mathbf{\sinf}} [H(Z|[E])_{\Phi(\rho_{ABE})} ]-\delta^{leak}_{EC}.
\end{equation}
The convex feasible set $\sinf$ is constrained by the parameter estimation Alice and Bob perform, as well as the reduced density matrix constraint for prepare-and-measure protocols \cite{Winick2018}. The error-correction cost $\delta^{leak}_{EC}$ can be observed directly and does not need to be optimized over. Using the shorthand $f$ for the convex objective function, our goal is to compute tight lower bounds on the following minimization,
\begin{equation}
\label{inf}
\min_{\rho \in\sinf} f(\rho).
\end{equation}

Tractable lower bounds on this infinite-dimensional optimization can be computed via the dimension reduction method \cite{Upadhyaya2021}. There are four steps to apply this method; we briefly summarize them here and give references to the relevant sections of Ref. \cite{Upadhyaya2021}. First choose a finite-dimensional subspace, represented by a projector $\Pi$ (Sec. IV A). Next, find a bound $W$ on the weight of $\rho$ outside this subspace (Sec. IV B). Third, determine a correction term $\Delta$ for the objective function $f$ (Sec. IV C). Finally, construct a finite-dimensional set $\sfin$ satisfying certain properties (Sec. IV D). The desired lower bound is then
\begin{equation}
\label{drm}
\min_{\rhot \in\sfin} f(\rhot)- \Delta(W) \leq \min_{\rho \in\sinf} f(\rho),
\end{equation}
where the finite-dimensional optimization can be solved numerically and the correction term is computed analytically. Tildes denote operators that are subnormalized.

\subsection{Correction Term}\label{corterm}
Intuitively, the correction term limits how much the function $f$ can increase under projection. Formally, it satisfies the following property,
\begin{equation}
\label{ucdupdefn}
\Tr(\rho \Pin) \leq W \implies f(\Pi \rho \Pi) - f(\rho) \leq  \Delta(W), \quad \forall \rho \in \sinf.
\end{equation}
In this case, we say that $f$ is \emph{uniformly close to decreasing under projection} (UCDUP) on $\sinf$, with correction term $\Delta$. The correction term we derive will apply on the set of all density operators, $\densityH[\Hinf]$, and for any choice of projection $\Pi$, so it can be applied to any QKD protocol.

\subsection{Postprocessing Map}
As we have noted, the postprocessing map $\Phi$ is a quantum-to-classical channel from $AB$ to $ZC$. It follows that $\Phi$ can be realized by a measurement. That is, $\Phi$ has the form 
\begin{equation}
\label{postmappovm}
\Phi(\rho_{ABE}) = \sum_{\substack{z \in \keypriv\\c \in \scpub}} \dyad{z}_Z \otimes \dyad{c}_C \otimes \Tr_{AB} \left[ (P^{z,c}_{AB} \otimes \id_E) \rho_{ABE}\right] ,
\end{equation}
where $\{ P^{z,c}_{AB} \}_{\substack{z \in \keypriv\\c \in \scpub}}$ is some positive operator-valued measure (POVM), over the alphabet of key symbols $\keypriv$ and public announcements $\scpub$ \footnote{As discussed in Ref. \cite{Upadhyaya2021}, the discard symbol $\perp$ is not included in the set of key symbols.}. For simplicity, we re-index the POVM by $k \in \skey\equiv \keypriv \times \scpub$. 

\section{Results}
We first introduce three lemmas which will be needed to prove our main theorem, which is a tight correction term. The first lemma is a continuity bound for conditional entropy in terms of trace distance. The second lemma lets us consider  the dephased state instead of the projected one; dephased means the off-diagonal blocks with respect to the projector and its complement are zeroed out. The third lemma provides a bound on the trace norm of a specific form of operator, which arises in the proof of the main theorem and is related to Eve's conditional states. 

\begin{lemma}[From Ref. \cite{Upadhyaya2021}]
\label{ucduplemma} 
Let $\mathcal{H}_A$ and $\mathcal{H}_B$ be two Hilbert spaces, where the dimension of $\mathcal{H}_A$ is $\abs{A}$ while $\mathcal{H}_B$ can be infinite-dimensional. Let $\rhot_{AB}, \sigmat_{AB}\in\tilde{D}( \mathcal{H}_A \otimes  \mathcal{H}_B)$ be two subnormalized, classical-quantum states with $\Tr(\rhot_{AB})\geq\Tr(\sigmat_{AB})$. If $\frac{1}{2}\tnorm{\rhot_{AB}-\sigmat_{AB}}\leq\epsilon$, then
\begin{equation}
H(A|B)_{\sigmat_{AB}} - H(A|B)_{\rhot_{AB}} \leq  \e \log_2\abs{A} +(1+\e)h\left(\frac{\e}{1+\e}\right),
\end{equation}
where $h(x)$ is the binary entropy function.
\end{lemma}
\begin{proof}
See Appendix A of Ref. \cite{Upadhyaya2021}.
\end{proof}

Define $\Xi_{AB}$ to be a dephasing channel associated with the projector $\Pi$ and its complement $\Pin$ as
\begin{equation}\label{dephasingchannel}
\Xi_{AB}(\rho)\equiv\Pi\rho\Pi+\Pin \rho \Pin.
\end{equation}
\begin{lemma}
\label{ictdephaslemma}
For any state $\rho_{AB}$, $f(\Pi \rho_{AB} \Pi) \leq H(Z|[E])_{\Phi(\Xi (\rho_{ABE}) )} $
\end{lemma}
\begin{proof}
Expanding definitions, we have that
\begin{align}
f(\Pi \rho_{AB} \Pi)&=H(Z|[E])_{\Phi(\Pi \rho_{ABE} \Pi)}\\ 
&\leq H(Z|[E])_{\Phi(\Pi \rho_{ABE} \Pi)}+H(Z|[E])_{\Phi(\Pin \rho_{ABE} \Pin)} \label{ictnonnegcq}\\
&\leq H(Z|[E])_{\Phi(\Pi \rho_{ABE} \Pi)+\Phi(\Pin \rho_{ABE} \Pin)}  \label{ictconcav}\\
&=H(Z|[E])_{\Phi(\Pi \rho_{ABE} \Pi+\Pin \rho_{ABE} \Pin)} \label{ictlinear}\\
&=H(Z|[E])_{\Phi(\Xi (\rho_{ABE}) )}.
\end{align}
Line $\eqref{ictnonnegcq}$ follows because the second term is the conditional entropy of a classical-quantum state and thus nonnegative, $\eqref{ictconcav}$ follows because conditional entropy is concave, and $\eqref{ictlinear}$ follows simply because the map $\Phi$ is linear.
\end{proof}

\begin{lemma}
\label{icttracenormlemma}
Let $P$ be a POVM element. With respect to a projection $\Pi$ and its complement $\Pin$, write $P$ as a block matrix
\begin{equation}
P =\begin{pmatrix}
\Pi P \Pi &  \Pi P \Pin \\
\Pin P \Pi & \Pin P \Pin \\
\end{pmatrix} \equiv
\begin{pmatrix}
A &  B \\
\adj{B} & D \\
\end{pmatrix}.
\end{equation}
Define $H = \begin{pmatrix}0 &  B \\\adj{B} & 0 \\\end{pmatrix}$ as the off-diagonal portion of $P$.

Let $\rho$ be a state, and define two new states corresponding to the normalized on-diagonal blocks of $\rho$: $\rho^\Pi = \frac{\Pi \rho \Pi}{\Tr( \rho \Pi)}$ and $\rho^{\Pin} = \frac{\Pin \rho \Pin}{\Tr( \rho \Pin)}$. Define the measurement probabilities $r = \Tr(\rho^\Pi P)$ and $s = \Tr(\rho^{\Pin} P)$. Let $W \geq \Tr(\rho \Pin)$.

It holds that 
\begin{equation}
\tnorm{\rrho H \rrho} \leq (r+s) \sqrt{W} \ \snorm{\sqrt{A}^g B \sqrt{D}^g}.
\end{equation}
\end{lemma}
Here  $\norm{\cdot}_p$ denotes the Schatten $p$-norm and $(\cdot)^g$ the generalized inverse. Of particular importance to us will be the trace norm ($p=1$) and the spectral norm ($p=\infty$). The generalized inverse is defined as the inverse of an operator on its support, so that $A^g A = \Pi_{supp(A)}$.

\begin{proof}
The trace norm can be expressed as a semidefinite program (SDP) in the following manner  \cite{Watrous2018},
\begin{IEEEeqnarray*}{uC;L}
$\tnorm{\rrho H \rrho} \  = \ \displaystyle{\minimize_{X}}$& &\Tr X \\ \IEEEyesnumber
subject to:& 	&X \geq \rrho H \rrho \\
&	& X \geq -\rrho H \rrho \label{icttnsdp1} \\ 
&	&X\geq 0.
\end{IEEEeqnarray*}

For any feasible $X$, $ \Pi_{\im(\rho)} X \Pi_{\im(\rho)}$ is still feasible, and can only decrease the value of the objective function ($\Pi_{\im(\rho)}$ is the projection onto the image of $\rho$). We can thus assume WLOG that $X=\Pi_{im(\rho)} X \Pi_{im(\rho)}$. 

Then, $X$ can be expressed as $\rrho R \rrho$ for some $R$. This lets us rewrite the SDP in Eq. \eqref{icttnsdp1} as
\begin{IEEEeqnarray*}{uC;L}
$\tnorm{\rrho H \rrho} \  = \ \displaystyle{\minimize_{R}}$& &\Tr (\rho R) \\ \IEEEyesnumber
subject to:& 	&\rrho R \rrho \geq \rrho H \rrho \\  
&	& \rrho R \rrho \geq -\rrho H \rrho \label{icttnsdp2} \\ 
&	&R\geq 0.
\end{IEEEeqnarray*}
As this is a minimization, any feasible guess leads to an upper bound. To show $R$ is feasible, it suffices to show $R$ is positive and satisfies $R\geq \pm H$.

The remainder of the proof consists of three steps. We first make a guess for a feasible $R$. We then prove that it is indeed feasible. Finally, we calculate the corresponding value of the objective function.

Recall that $P$ is positive as it is a POVM element. In terms of the block matrix characterization with respect to the projectors $\Pi$ and $\Pin$, we have that
\begin{equation}
\begin{pmatrix}
A &  B \\
\adj{B} & D \\
\end{pmatrix}  \geq 0.
\end{equation}
By Theorem IX.5.9 of Ref. \cite{Bhatia1996}, the above holds if and only if $\sqrt{A}^g B \sqrt{D}^g \equiv K$ is a contraction, i.e. $\snorm{K}\leq1$ \footnote{In Ref. \cite{Bhatia1996}, this theorem is proven for finite-dimensional matrices where $B$ is a square block. However, nothing precludes the proof from applying in infinite dimensions and with rectangular blocks.}. Note $\snorm{K}$ ranges from $0$ to $1$ and quantifies how close to block-diagonal the POVM element is. In particular, $\snorm{K}=0$ when the POVM is exactly block-diagonal, since then $B=0$.

We now specify our guess to be 
\begin{equation}
\label{Rfeasible}
R= a\Pi P \Pi + b \Pin P \Pin,
\end{equation}
with the constants $a= \snorm{K} \sqrt{W}$ and $b=\frac{\snorm{K} }{\sqrt{W}}$ (assuming $W\neq0$, as the lemma follows immediately for $W=0$).

Let us verify that this guess is feasible. Since $R$ is manifestly positive, it suffices to show $R\pm H\geq 0$. Written in terms of block matrices, this condition is equivalent to 
\begin{equation}
R\pm H =\begin{pmatrix}
aA & \pm B \\
\pm \adj{B} & b D \\
\end{pmatrix} \geq 0.
\end{equation}

Again by Theorem IX.5.9 of Ref. \cite{Bhatia1996}, this condition is satisfied if and only if $\sqrt{aA}^g  (\pm B)  \sqrt{bD}^g $ is a contraction. Noting that $\sqrt{ab}=\snorm{K}$, this simplifies as
\begin{align}
\sqrt{aA}^g  (\pm B) \sqrt{bD}^g &= \frac{1}{\sqrt{ab}} \sqrt{A}^g  (\pm B) \sqrt{D}^g\\
&= \frac{\pm K}{\snorm{K}}.
\end{align}
The operator on the last line clearly has unit norm so is a contraction. Thus, it follows that $R\pm H\geq0$ and $R$ is a feasible guess. 

The objective function value is
\begin{equation}
\Tr(\rho R) = a \Tr(\rho \Pi P \Pi) + b \Tr(\rho \Pin P \Pin).
\end{equation}
Recall the measurement probabilities $r = \Tr(\rho^\Pi P)$ and $s = \Tr(\rho^{\Pin} P)$ introduced above. The first term can be upper bounded as
\begin{align}
a \Tr(\rho \Pi P \Pi) &= a \Tr(\Pi \rho \Pi \ \Pi P \Pi)\\
&= a \Tr(\Pi \rho \Pi) \Tr( \rho^\Pi  P )\\
&\leq a \Tr( \rho^\Pi  P )\\
&= r \snorm{K} \sqrt{W} ,
\end{align}
where in the second line we pull out the normalization of $\rho^\Pi$.

Similarly for the second term,
\begin{align}
b \Tr(\rho \Pin P \Pin) &= b \Tr(\Pin \rho \Pin\ \Pin P \Pin)\\
&= b \Tr(\Pin \rho \Pin) \Tr( \rho^{\Pin} P )\\
&\leq b W \Tr( \rho^{\Pin} P)\\
&= s \snorm{K} \sqrt{W}.
\end{align}

Thus, the feasible choice of $R$ in Eq. \eqref{Rfeasible} leads to the following upper bound on Eq. \eqref{icttnsdp2},
\begin{equation}
\tnorm{\rrho H \rrho}  \leq (r+s) \snorm{K}  \sqrt{W},
\end{equation}
and the proof is complete.
\end{proof}

We now state the theorem for the improved correction term. 
\begin{thm}
\label{ictucdupthm}
Consider the QKD objective function $f(\rho_{AB})=H(Z|[E])_{\Phi(\rho_{ABE})}$, with the map $\Phi$ defined by a POVM $\{ P_k \}_{k \in \skey}$ (see Eq. \eqref{postmappovm}). With respect to an arbitrary projection $\Pi$, write each $P_k$ as a block matrix
\begin{equation}
P_k =\begin{pmatrix}
\Pi P_k \Pi &  \Pi P_k \Pin \\
\Pin P_k \Pi & \Pin P_k \Pin \\
\end{pmatrix} \equiv
\begin{pmatrix}
A_k &  B_k \\
\adj{B_k} & D_k \\
\end{pmatrix}.
\end{equation}

For this projection $\Pi$, the QKD objective function $f$ is UCDUP on $\normdensityH[\Hinf]$ with correction term
\begin{equation}
\label{qkdictucdup}
\Delta(W)= c \sqrt{W} \log_2 \abs{Z} +  \left(1+c\sqrt{W} \right) h\left(\frac{c\sqrt{W}}{1+c\sqrt{W}}\right),
\end{equation} 
where $\abs{Z}=\abs{S_Z}$ is the dimension of the key map register and 
\begin{equation}
c= \max_{k \in \skey} \ \snorm{\sqrt{A_k}^g B_k \sqrt{D_k}^g} .
\end{equation}
\end{thm}

\begin{proof}[Proof of Theorem \ref{ictucdupthm}]
As per the definition of UCDUP, let $\rho_{AB}\in\normdensityH[\Hinf]$ be a state satisfying $\Tr(\rho_{AB} \Pin) \leq W$.

We first bound the trace distance between $\Phi(\Xi (\rho_{ABE}) )$ and $\Phi(\rho_{ABE} )$. We have
\begin{align}
\Phi(\Xi (\rho_{ABE})) &= \sum_k \dyad{k}_K \otimes \Tr_{AB} \left[ (P^k_{AB} \otimes \id_E) \Xi(\rho_{ABE})\right] \\
&=  \sum_k \dyad{k}_K \otimes \Tr_{AB} \left[(\Xi(P^k_{AB}) \otimes \id_E) \rho_{ABE}\right] \label{adjointpovm}
\end{align}
since the dephasing channel $\Xi$ is self-adjoint. Note that $\{ \Xi(P^k) \}$ is also a POVM, as the channel is positive and unital. In writing $\Phi(\Xi (\rho_{ABE}))$ in this manner, we are comparing the effect of two different channels on the same input state, instead of the same channel on two different inputs. 

Since the trace norm is additive over blocks corresponding to orthogonal subspaces, we have
\begin{align}
&\tnorm{\Phi(\rho_{ABE} ) - \Phi(\Xi (\rho_{ABE}) )}\\
&= \tnorm{\sum_k \dyad{k}_K \otimes \Tr_{AB} \left[\left(\left[P^k_{AB}-\Xi(P^k_{AB})\right] \otimes \id_E\right) \rho_{ABE}\right]}\\
& = \sum_k \tnorm{\Tr_{AB} \left[\left(\left[P^k_{AB}-\Xi(P^k_{AB})\right] \otimes \id_E\right) \rho_{ABE}\right]}.\label{tempeqn4}
\end{align}

To proceed, we find a more useful form for Eve's conditional states. Recall that Eve's register $E$ purifies $\rho_{AB}$. We can thus assume that Eve's register has the same dimension as Alice and Bob's. That is, $\Hspace_E=\Hspace_{AB}$. There then exists a bijective isometry $V: \Hspace_{AB} \rightarrow \Hspace_{E}$.
(To construct such an isometry, simply choose a basis $\ket{i}_{AB}$ for $\Hspace_{AB}$ and a basis $\ket{i}_E$ for $\Hspace_E$, and define $V\ket{i}_{AB}=\ket{i}_E$.) Via the vectorization mapping, it can easily be shown that 
\begin{equation}
\label{veccond}
\Tr_{AB} \left(\left(P^k_{AB} \otimes \id_E\right) \rho_{ABE} \right)= V\left( \sqrt{\rho_{AB}} P^k_{AB} \sqrt{\rho_{AB}} \right)\adj{V},
\end{equation}
and similarly for $\Xi(P^k_{AB})$.

Applying this identity to Eq. \eqref{tempeqn4}, we have
\begin{align}
&\tnorm{\Phi(\rho_{ABE} ) - \Phi(\Xi (\rho_{ABE}) )}\\
&= \sum_k \tnorm{ V \left( \sqrt{\rho_{AB}} \left[P^k_{AB}-\Xi(P^k_{AB})\right] \sqrt{\rho_{AB}} \right) \adj{V} }\\
&=\sum_k \tnorm{ \left( \sqrt{\rho_{AB}} \left[P^k_{AB}-\Xi(P^k_{AB})\right] \sqrt{\rho_{AB}} \right)}\\
&= \sum_k \tnorm{ \left( \sqrt{\rho_{AB}} \left[ \Pi P^k_{AB} \Pin + \Pin P^k_{AB} \Pi\right] \sqrt{\rho_{AB}} \right)}. \label{tempeqn3}
\end{align}
In keeping with our previous notation, we define $\rho^\Pi = \frac{\Pi \rho \Pi}{\Tr( \rho \Pi)}$ and $\rho^{\Pin} = \frac{\Pin \rho \Pin}{\Tr( \rho \Pin)}$, as well as the probability distributions $r(k) = \Tr(\rho^\Pi P_k)$ and $s(k) = \Tr(\rho^{\Pin} P_k)$. By Lemma \ref{icttracenormlemma} we have
\begin{align}
\sum_k &\tnorm{ \left( \sqrt{\rho_{AB}} \left[ \Pi P^k_{AB} \Pin + \Pin P^k_{AB} \Pi\right] \sqrt{\rho_{AB}} \right)} \\
&\leq \sum_k (r(k)+s(k)) \sqrt{W}  \snorm{\sqrt{A_k}^g B_k \sqrt{D_k}^g} \\
&= \sqrt{W} \bigg( \sum_k r(k) \snorm{\sqrt{A_k}^g B_k \sqrt{D_k}^g}  \nonumber \\
&\quad \quad \quad \quad\quad + \sum_k s(k)  \snorm{\sqrt{A_k}^g B_k \sqrt{D_k}^g}  \bigg)  \\
&\leq \sqrt{W} \bigg( \max_k \snorm{\sqrt{A_k}^g B_k \sqrt{D_k}^g}  \label{tempeqn} \nonumber\\
& \quad \quad \quad \quad \quad + \max_k \snorm{\sqrt{A_k}^g B_k \sqrt{D_k}^g}  \bigg)  \\
&= 2 \sqrt{W} \bigg(  \max_k \snorm{\sqrt{A_k}^g B_k \sqrt{D_k}^g}  \bigg),
\end{align}
where in Eq. \eqref{tempeqn} we used the fact that the sums over $r(k)$ and over $s(k)$ are both convex combinations of $\snorm{\sqrt{A_k}^g B_k \sqrt{D_k}^g} $, and hence upper-bounded by the largest of these terms. Re-inserting this bound in Eq. \eqref{tempeqn3}, we have
\begin{align}
&\tnorm{\Phi(\rho_{ABE} ) - \Phi(\Xi (\rho_{ABE}) )} \nonumber \\&\quad\quad\quad \leq 2 \sqrt{W}  \left( \max_k \snorm{\sqrt{A_k}^g B_k \sqrt{D_k}^g}  \right)
\end{align}
or
\begin{align}
&\frac{1}{2} \tnorm{\Phi(\rho_{ABE} ) - \Phi(\Xi (\rho_{ABE}) )} \nonumber \\&\quad\quad\quad  \leq \sqrt{W}  \left( \max_k \snorm{\sqrt{A_k}^g B_k \sqrt{D_k}^g}  \right).
\end{align}

Letting $c=  \max_k \snorm{\sqrt{A_k}^g B_k \sqrt{D_k}^g} $, by the continuity bound in Lemma \ref{ucduplemma}, we have 
\begin{align}
&H(Z|[E])_{\Phi(\Xi (\rho_{ABE}) )} - f(\rho_{AB}) \nonumber  \\
&\quad \quad \quad \leq c\sqrt{W} \log_2 \abs{Z} +\left(1+c\sqrt{W}\right) h\left(\frac{c\sqrt{W}}{1+c\sqrt{W}}\right).
\end{align}
By Lemma \ref{ictdephaslemma}, and the definition of UCDUP (Eq. \eqref{ucdupdefn}), the theorem statement follows.
\end{proof}

Note that $c$ is zero when all the POVM elements are block-diagonal, and increases to a maximum of 1 as the off-diagonal blocks of any POVM element get larger. It can be interpreted as a normalized measure of how well the POVM elements commute with the projection. This also indicates that a good choice of projection is one for which all POVM elements are close to block-diagonal. 

In Figure \ref{fig1}, we plot the improved correction term as a function of $c$, compared to the two cases we had before, for $\abs{Z}=4$. The new term smoothly interpolates between the commuting and maximally non-commuting cases. We see that if $c$ is small, the improved correction term reaches smaller values at larger $W$ than the old correction term. This in turn reduces the computational overhead, as it allows us to further reduce the dimension of the finite-dimensional subspace in our key rate calculation, even though this comes with an increase in the weight $W$ outside that subspace.

\begin{figure}
\centering
\includegraphics[width=\linewidth]{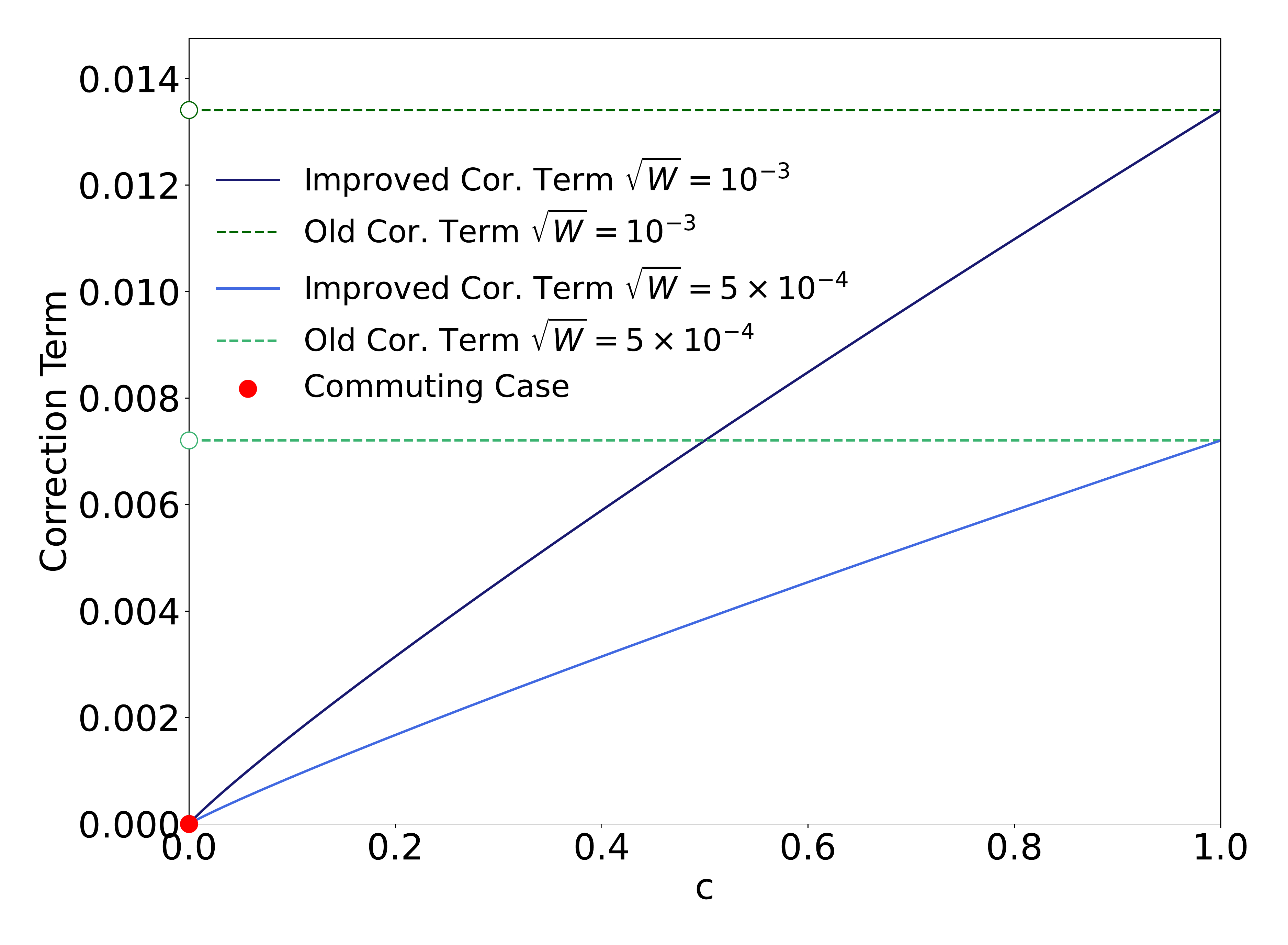}
\caption{Improved correction term compared to the old piecewise correction terms in the commuting and non-commuting case: $\abs{Z}=4$.}
\label{fig1}
\end{figure}

To calculate the improved correction term for a specific protocol, $c$ needs to be upper-bounded. If the POVM elements are rank-one, then $c$ can be calculated. For more complicated POVM elements, analytic tools to bound $c$ would be an important area for future research. To estimate $c$ numerically, we can calculate $\snorm{\sqrt{A_k}^g B_k \Pi_C \sqrt{D_k}^g \Pi_C}$, where $\Pi_C$ is a projector onto a subspace containing $\Pi$. By increasing the dimension of $\Pi_C$ until the approximations to $c$ seem to converge, we can estimate $c$ \cite{Upadhyaya2021M}.

\section{Conclusion}
In summary, we have determined a tighter form for the correction term in the dimension reduction method. This new correction term is applicable to general QKD protocols, and is small when the protocol has POVM elements which are close to block-diagonal. This result is conceptually interesting because it provides a unified perspective, interpolating between the cases of zero and maximal off-diagonal blocks. It is also practically relevant, because having a smaller correction term allows one to achieve the same key rate with a larger value of $W$, which corresponds to solving the finite-dimensional optimization in fewer dimensions. Due to the computational limitations of numerical SDP solvers, this can provide a significant numerical advantage. 

Interesting directions for future work include applying this correction term to specific protocols and finding regimes where it gives the most significant improvement.

\begin{acknowledgments}
The Institute for Quantum Computing is supported in part by Innovation, Science, and Economic Development Canada. This research has been supported by NSERC under the Discovery Grants Program, Grant No. 341495, and under the Collaborative Research and Development Program, Grant No. CRDP J 522308-17. Financial support for this work has been partially provided by Huawei Technologies Canada Co., Ltd.
\end{acknowledgments}

\bibliographystyle{modrevtex}
\bibliography{mainbib.bib}

\end{document}